\documentclass[a4paper,conference]{IEEEtran}

\usepackage{cite}
\usepackage{amsmath,amssymb,amsfonts}
\allowdisplaybreaks
\usepackage{nccmath}
\usepackage{algorithm}
\usepackage{algorithmic}
\usepackage{graphicx}
\usepackage{textcomp}
\usepackage{xcolor}
\usepackage{float}
\usepackage{subfigure}
\usepackage{subeqnarray}
\usepackage{adjustbox}
\PassOptionsToPackage{hyphens}{url}\usepackage{hyperref}
\usepackage[justification=centering]{caption}
 
\usepackage{graphicx}  

\usepackage{caption}
 
\DeclareCaptionLabelFormat{lc}{{#1}~#2}
\usepackage{pifont} 

\makeatletter
\def\endthebibliography{%
  \def\@noitemerr{\@latex@warning{Empty `thebibliography' environment}}%
  \endlist
}
\makeatother

 \usepackage{epsfig,endnotes}

\usepackage{floatpag,enumitem}

\usepackage{relsize}
 
 \usepackage{tikz}

\usepackage{multirow}
\usepackage{setspace} 
\usepackage{mathrsfs}
\usepackage{bbm}
\usepackage{pifont}
\usepackage{fancyhdr}
\usepackage{amsfonts}

\usepackage{amsmath, amsthm}

\usepackage{tabularx}
\usepackage{listings}
\usepackage{graphicx}
\usepackage{amssymb,booktabs}
\usepackage{array}
\usepackage{cases}

 \usepackage{supertabular}
\usepackage{mathrsfs}

\usepackage{psfrag}
\usepackage{bbding}

\usepackage{float}

\usepackage{amsbsy}

\makeatletter
\providecommand{\leftsquigarrow}{%
  \mathrel{\mathpalette\reflect@squig\relax}%
}
\newcommand{\reflect@squig}[2]{%
  \reflectbox{$\m@th#1\rightsquigarrow$}%
}
\makeatother

\usepackage{algorithm}
 \usepackage{algorithmic}

\makeatletter
\newcommand{\newalgname}[1]{%
  \renewcommand{\ALG@name}{#1}%
}

\newcommand {\C} {{\rm I\kern-5.5pt C}}

\def\centerhack#1{\hbox to 0pt{\hss\footnotesize #1\hss}}
\def\centerhackn#1{\hbox to 0pt{\hss #1\hss}}
\def\dchack#1{\vbox to 0pt{\vss{\hbox to 0pt{\hss#1\hss}}\vss}}

\usepackage{etoolbox}
\AtBeginEnvironment{align}{\setcounter{subeqn}{0}}
\newcounter{subeqn} %

\usetikzlibrary{positioning}

\newcounter{mysub}

\newtheorem{lem}{Lemma}

\newtheorem{thm}{Theorem}

\newtheorem*{proposition1.1}{Proposition 1.1}
\newtheorem*{proposition1.2}{Proposition 1.2}
\newtheorem*{proposition1.3}{Proposition 1.3}
\newtheorem*{proposition2.1}{Proposition 2.1}
\newtheorem*{proposition2.2}{Proposition 2.2}
\usepackage{subfigure}

\usepackage[top=0.75in, bottom=1in, left=0.625in, right=0.625in,footskip=.3in]{geometry}

\makeatletter
\newcommand{\linebreakand}{%
  \end{@IEEEauthorhalign}
  \hfill\mbox{}\par
  \mbox{}\hfill\begin{@IEEEauthorhalign}
}
\makeatother

\begin{document}

\title{Resource Allocation in Large Language Model Integrated 6G Vehicular Networks}

\author{\IEEEauthorblockN{Chang Liu and Jun Zhao\\}
\IEEEauthorblockA{{Nanyang Technological University, Singapore}\\
liuc0063@e.ntu.edu.sg, junzhao@ntu.edu.sg }
}
\maketitle

\thispagestyle{fancy}
\pagestyle{fancy}
\lhead{This paper appears in the 2024 IEEE 99th Vehicular Technology Conference (VTC).}
\cfoot{\thepage}
\renewcommand{\headrulewidth}{0.4pt}
\renewcommand{\footrulewidth}{0pt}

\begin{abstract}
In the upcoming 6G era, vehicular networks are shifting from simple Vehicle-to-Vehicle (V2V) communication to the more complex Vehicle-to-Everything (V2X) connectivity. At the forefront of this shift is the incorporation of Large Language Models (LLMs) into vehicles. Known for their sophisticated natural language processing abilities, LLMs change how users interact with their vehicles. This integration facilitates voice-driven commands and interactions, departing from the conventional manual control systems.
However, integrating LLMs into vehicular systems presents notable challenges. The substantial computational demands and energy requirements of LLMs pose significant challenges, especially in the constrained environment of a vehicle. Additionally, the time-sensitive nature of tasks in vehicular networks adds another layer of complexity.
In this paper, we consider an edge computing system where vehicles process the initial layers of LLM computations locally, and offload the remaining LLM computation tasks to the Roadside Units (RSUs), envisioning a vehicular ecosystem where LLM computations seamlessly interact with the ultra-low latency and high-bandwidth capabilities of 6G networks.
To balance the trade-off between completion time and energy consumption, we formulate a multi-objective optimization problem to minimize the total cost of the vehicles and RSUs.
The problem is then decomposed into two sub-problems, which are solved by sequential quadratic programming (SQP) method and fractional programming technique.
The simulation results clearly indicate that the algorithm we have proposed is highly effective in reducing both the completion time and energy consumption of the system. 
\end{abstract}


\section{Introduction}
With the advent of 6G networks, the integration of advanced computational systems, specifically LLMs and vehicular networks, marks a significant step forward.
It is primarily motivated by the goal of enhancing the driving experience in anticipation of the 6G network era. These advanced AI systems introduce a new era of driver-vehicle interaction, moving away from conventional manual controls towards more intuitive, voice-driven interfaces~\cite{wen2024road}. By enabling natural language communication, LLMs significantly reduce cognitive distractions, allowing drivers to maintain focus on the road. 
This improvement helps make driving safer and also makes cars easier for more people to drive, paving the way for inclusive mobility in the 6G era.

The implications of LLM integration play a vital role in the development of autonomous driving technologies~\cite{cui2023drivellm}, which are expected to be further empowered by 6G networks. 
LLMs can aid in real-time decision-making and understanding complex road scenarios, which are essential for the safe operation of autonomous vehicles in the 6G-connected environment~\cite{cui2024survey}. Additionally, these models offer significant contributions to vehicle maintenance and diagnostics, predicting and notifying maintenance needs, thereby preventing breakdowns and extending vehicle lifespan. The economic and environmental benefits are also noteworthy, with efficient route planning and predictive maintenance contributing to reduced fuel consumption and lower emissions. In essence, the integration of LLMs aims not only to enhance user experience and safety but also to propel the automotive industry towards a more autonomous and sustainable future.

The integration of LLMs into vehicular technology, while offering numerous benefits, presents a significant challenge in terms of computational resource requirements. 
These sophisticated models demand extensive processing power to function effectively, a demand that becomes particularly pronounced in the context of vehicular systems. 
Operating such advanced AI models requires substantial computational resources, which can be both time-consuming and energy-intensive~\cite{cai2023efficient}, a challenge that will be exacerbated in the context of 6G-enabled vehicular networks. 
In a vehicular environment, where the availability of computational resources is inherently limited, and energy efficiency is paramount, this high demand for resources poses a unique challenge.

The energy consumption associated with running these models is a critical concern in 6G-enabled vehicular networks, where energy efficiency directly impacts the vehicle's range and overall performance. 
To address this, one common solution is leveraging edge computing by deploying RSUs~\cite{ning2020intelligent,wang2018offloading}. 
Edge computing can facilitate immediate decision-making by empowering the network edge with intelligence~\cite{zhang2020data}. 
Within vehicular networks, edge computing enables vehicles to leverage the higher computational power of these RSUs, reducing the processing burden on the vehicle itself. This not only alleviates the strain on the vehicle's limited computational resources but also helps in conserving energy, thereby enhancing the overall efficiency and sustainability of vehicular systems. Edge computing with RSUs thus emerges as a pivotal strategy in the optimization of resource allocation, ensuring that the integration of LLMs in vehicles is both technically feasible and energy-efficient.

In this paper, we explore an edge computing vehicular framework where vehicles are equipped to handle the initial several layers of LLM computations locally, while offloading the more intensive remaining LLM computation tasks to RSUs. 
This approach aims to leverage the computational capabilities of both vehicles and network edges effectively. 
To achieve an optimal balance between the completion time of tasks and the overall energy consumption, we develop a multi-objective optimization framework tailored to the requirements of 6G-enabled vehicular networks.
This framework is specifically designed to minimize the cumulative cost incurred by both vehicles and RSUs in processing these computational tasks.
The intricate nature of the optimization problem prompts us to divide it into two sub-problems for more efficient resolution. 
The first sub-problem is addressed using the Sequential Quadratic Programming (SQP) method. 
The second sub-problem is tackled using the fractional programming technique, which is particularly adept at dealing with ratio-based optimization problems.

The main contributions are summarized as follows:
\begin{itemize}
    \item We introduce a novel concept for vehicular networks, integrating LLMs to enhance their capabilities. In this proposed system, the vehicles are configured to process the initial layers of LLM computations locally. By leveraging edge computing,  more demanding and resource-intensive portions of LLM computations are offloaded to RSUs.
    \item We propose an optimization strategy for vehicular networks using LLMs, focusing on optimizing the number of LLM layers offloaded to RSUs, the transmission power and GPU frequency of vehicles, and the bandwidth and GPU frequency allocation for RSUs, taking into account the unique requirements of 6G-enabled vehicular networks.
    This novel approach uniquely balances computational loads and energy consumption across the network, improving the system's efficiency.
    \item To tackle the non-convex problem presented by the intricate nature of the optimization challenge, we partition it into two sub-problems. The first sub-problem is addressed using the Sequential Quadratic Programming (SQP) method. 
    For the second sub-problem, we employ the fractional programming technique.
    This two-pronged strategy allows for a more effective solution to the complex optimization problem.
\end{itemize}

\section{System Model}
In this section, we delve into the system model in detail.
First, we explore the methodology for calculating both time and energy consumption associated with the Large LLM within vehicular networks. 
Following that, we then break down the computation models for two key components: the vehicle computation model and the RSU computation model.

\subsection{Time and Energy Calculation}
Firstly, we describe how to calculate the time and energy consumption for the LLM.
Following the derivation from Narayanan \textit{et al.}~\cite{narayanan2021efficient}, we only consider the matrix multiplications which are the main contributors to computation cost.
Assume 
the dimensionality of the hidden states is $h$; 
the length of input tokens is $d_n$, and batch size is $B$.
A transformer layer is built from two primary components: the attention mechanism and a feed-forward network. 
The attention mechanism's computational demands arise from transformations of the key, query, and value ($6Bd_nh^2$ operations), the attention matrix's computation ($2Bd_n^2h$ operations), the application of attention to values ($2Bd_n^2h$ operations), and a subsequent linear transformation ($2Bd_nh^2$ operations). 
The feed-forward network modifies the hidden size, first expanding it to $4h$ and then contracting it back to $h$, which requires $16Bd_nh^2$ FLOPs.
Combining these, the forward pass of each transformer layer requires $24Bd_nh^2+4Bd_n^2h$ FLOPs. 
The backward pass doubles this computational cost because it calculates gradients for both the input and weight tensors.
Let $\psi(d_n)$ be the FLOPs per token required by vehicle $n$ for forward pass per transformer layer where $\psi(d_n) = 24Bd_nh^2+4Bd_n^2h$.

\subsection{Vehicle Computation Model}
When vehicles undertake the task of executing one transformer layer locally, the computation delay is expressed as follows:
\begin{align}
    T_n^{cmp}=\frac{\psi(d_n)}{f_n^\text{FL}},
\end{align}
where $f_n^\text{FL}$ is the number of FLOPs per cycle per core of the GPU at vehicle $n$, 
and $f_n^\text{FL}=f_n C_n^{V} D_n^{V}$ where $f_n$ is the GPU frequency of vehicle $n$, $C_n^{V}$ is the number of cores of the GPU at vehicle $n$ and $D_n^{V}$ is the number of FLOPs per cycle per core of the GPU at vehicle $n$.
According to Zeng \textit{et al.}~\cite{zeng2021energy} and Eyerman \textit{et al.}~\cite{eyerman2011fine}, the computing power has a cubic relationship with the GPU frequency where $\texttt{power}=\kappa_1f_n^3$ and $\kappa_1$ is coefficient $[\text{in Watt/(cycle/s)}^3]$ conditional on the chip architecture.
Thus, the energy consumption of local computation for one transformer layer can be formulated by:
\begin{align}
    E_n^{cmp} = \texttt{power} \times T_n^{cmp}=\frac{\kappa_1f_{n}^2\psi(d_n)}{ C_n^{V} D_n^{V}}. 
\end{align}

After local computing, the vehicles will send the intermediate results and labels to the RSUs to perform the remaining computation. 
Let $\chi_{n,m}$ denote the vehicle-to-RSU association, where $\chi_{n,m}=1$ indicates vehicle $n$ choose RSU $m$ to perform the remaining computation.
Otherwise, $\chi_{n,m}=0$.
We consider frequency-division multiple access (FDMA) in this paper, where communication among vehicles and RSUs would not interfere.
Let $p_n$ be the transmission power of the user $n$. 
According to the Shannon-Hartley theory~\cite{cover1999elements}, the transmission rate between vehicle $n$ and RSU $m$ can be formulated as:
\begin{align}
    r_{n,m} = b_n\log_2(1+\frac{g_{n,m}p_{n}}{\sigma^2b_n}),
\end{align}
where $\sigma^2$ is the noise power, $b_n$ is the  bandwidth allocated to vehicle $n$.
Therefore, the energy consumption for vehicle $n$ to transmit data with size $d_n$ is:
\begin{align}
    E_{n}^{com} =  \frac{p_nd_n}{ \sum_{m \in \mathcal{M}}\chi_{n,m} r_{n,m}}.
\end{align}
Assume there are $\Upsilon$ transformer layers in total in the LLM model, and vehicle $n$ executes the first $\alpha_n$ transformer layers locally.
Therefore, the cost of the inference for first $\alpha_n$ transformer layers at vehicle $n$ is:
\begin{align}
    Cost_{n}^{V}=\alpha_n \cdot(\omega_t T_n^{cmp} + \omega_e E_n^{cmp})+\omega_eE_n^{com},
\end{align}
where $\omega_t$ is the weight parameter that indicates the preference for delay, and $\omega_e$ is the weight parameter that indicates the preference for energy consumption.
The sum of the weight parameters is $1$, i.e., $\omega_t+ \omega_e =1$.
\subsection{RSU Computation Model}
Consider $f_{n,m}$, which serves as the designated frequency of the GPU at a specific RSU $m$, that is allocated to a vehicle $n$.
The computational capability of this setup, particularly in terms of Floating Point Operations per Second (FLOPS), is determined by the following equation:
\begin{align}
    f_{n,m}^\text{FL}=f_{n,m} C_m^{R} D_m^{R}.
\end{align}
In this formula, $C_m^{R}$ represents the total count of GPU cores that are installed and operational at RSU $m$. Additionally, $D_m^{R}$ refers to the number of floating-point operations that can be performed in a single cycle by each of these cores at RSU $m$.

Moving on to the computational delay, we specifically focus on the time required for processing the forward pass of data pertaining to vehicle $n$ at RSU $m$. 
This delay is calculated by the following expression:
\begin{align}
    T^{cmp}_{n,m} = \frac{\psi(d_n)}{f_{n,m}^\text{FL}}=\frac{\psi(d_n)}{f_{n,m} C_m^{R} D_m^{R}}.
\end{align}
Accordingly, concerning the energy consumption of RSU $m$ for executing the inference task for vehicle $n$, the energy consumption is given by:
\begin{align}
    E_{n,m}^{cmp} =\frac{ \kappa_2 f_{n,m}^2 \psi(d_n)}{C_m^{R}D_m^{R}},\\[-13pt]\nonumber
\end{align}
where $\kappa_2$ is the coefficient depending on the chip architecture.
It's important to note that our analysis deliberately excludes the energy consumption associated with the downlink transmission from RSUs to vehicles. 
This decision stems from the understanding that RSUs typically possess significantly more robust power capabilities when compared to vehicles. 
Moreover, when comparing this aspect with the significant energy requirements involved in training LLMs, it becomes apparent that the energy consumed for transmission is relatively minor for the RSUs.

For each vehicle $n$, $\alpha_n$ layers out of $\Upsilon$ layers are computed locally, and thus $(\Upsilon-\alpha_n)$ layers are computed at the corresponding RSU.
Then, the cost to process the inference task for the  vehicles at the RSU $m$ is given by the weighted sum of delay and energy consumption:
\begin{align}
    Cost_{m}^{R}= \sum_{n \in \mathcal{N}}\chi_{n,m}(\Upsilon-\alpha_n) (\omega_t T_{n,m}^{cmp}+\omega_e E_{n,m}^{cmp}),
\end{align}
where $\omega_t$ represents the weighting factor that signifies the inclination towards delay, and $\omega_e$ denotes the weighting factor that signifies the inclination towards energy consumption.

\subsection{Problem Formulation}
With the computation and communication model above, we then formulate the joint optimization problem that aims to minimize the system's cost, by optimizing the following variables: the number of transformer layers that execute locally: \mbox{$\boldsymbol{\alpha}:=[\alpha_n|_{n \in \mathcal{N}}]$}, 
the transmission power of the vehicles: $\boldsymbol{p}:=[p_n|_{n \in \mathcal{N}}]$, the bandwidth allocation: \mbox{$\boldsymbol{b}:=[b_n|_{n \in \mathcal{N}}]$}, the vehicle's GPU frequency: $\boldsymbol{f^V}:=[f_{n}|_{n \in \mathcal{N}}]$ and the RSUs' GPU frequency allocation: $\boldsymbol{f^R}:=[f_{n,m}|_{n \in \mathcal{N}, m \in \mathcal{M}}]$:
\begin{subequations}\label{P1}
\begin{align}
    &\text{Problem~}\mathbb{P}_1:\nonumber\\
    &\min_{\boldsymbol{\alpha,p,b,f^V,f^R}} \sum_{n \in \mathcal{N}} Cost_{n}^{V} + \sum_{m \in \mathcal{M}} Cost_{m}^{R}\tag{\ref{P1}},\nonumber\\
    &\text{~~~s.t.~}\alpha_n \in \{1,2,\ldots,\Upsilon\},\forall n \in \mathcal{N},\label{P1_C_alpha}\\
    &\phantom{~~~s.t.~}p_n \leq p_n^{max},\forall n \in \mathcal{N},\label{P1_C_pn}\\
    &\phantom{~~~s.t.}\sum_{n \in \mathcal{N}} \chi_{n,m} b_n=b_m^{max},\forall m \in \mathcal{M},\label{P1_C_bn_2}\\
    &\phantom{~~~s.t.~} f_n \leq f_n^{max},\forall n \in \mathcal{N},\label{P1_C_f_n}\\
    &\phantom{~~~s.t.}\sum_{n \in \mathcal{N}}\chi_{n,m} f_{n,m}=f_m^{max},\forall m \in \mathcal{M}.\label{P1_C_fnm_2}
\end{align}
\end{subequations}


The formulated problem falls into the category of Mixed Integer Non-linear Programming (MINLP) problem.
This classification arises due to the inclusion of both integer-valued decision variables and non-linear terms involving products of variables, a combination that inherently induces non-convexity in the problem space.
The non-convex nature of this problem makes it especially challenging to solve because it cannot be addressed using standard optimization methods, which typically rely on the problem being convex.

\textbf{The roadmap to solve problem $\mathbb{P}_1$:}
In order to tackle the non-convex problem, we decompose the original into 2 sub-propblems.
Following this, we utilize Alternating Optimization (AO) by optimizing $\boldsymbol{\alpha,f^V,f^R}$ and $\boldsymbol{p,b}$ iteratively.
The two AO steps are described as follows:
\begin{enumerate}
\item In the first step of AO, we solve Sub-problem \uppercase\expandafter{\romannumeral1}.
Specifically, we fix $\boldsymbol{p,b}$ and utilize Sequential Quadratic Programming (SQP) method to optimize $\boldsymbol{\alpha,f^V,f^R}$ by transforming the non-convex problem into a convex one.
\item In the second step of AO, we solve Sub-problem \uppercase\expandafter{\romannumeral2}. In this sub-problem, given $\boldsymbol{\alpha,f^V,f^R}$, the method of fractional programming is adopted to facilitate the solution to $\boldsymbol{p,b}$.
\end{enumerate}
\section{Proposed Algorithm to Solve the Problem}
In this section, we present well-structured solutions to the two sub-problems.
\subsection{Sub-problem \uppercase\expandafter{\romannumeral1}}
In Sub-problem \uppercase\expandafter{\romannumeral1}, we fix $\boldsymbol{p,b}$ and optimize $\boldsymbol{\alpha,f^V,f^R}$.
The discrete variable $\alpha_n$ is difficult to handle.
Thus, we first relax $\alpha_n$ to continuous variables, which will be rounded back to the nearest integer later.
Thus, Sub-problem \uppercase\expandafter{\romannumeral1} is given by:
\begin{subequations}\label{sub-problem1}
\begin{align}
    &\text{Sub-problem \uppercase\expandafter{\romannumeral1}:}\nonumber\\&\min_{\boldsymbol{\alpha,f^V,f^R}} F(\boldsymbol{\alpha,f^V,f^R})\!=\!\!\sum_{n \in \mathcal{N}} \alpha_n \cdot(\omega_t T_n^{cmp} + \omega_e E_n^{cmp})+\nonumber\\&~~~~\sum_{m \in \mathcal{M}}\sum_{n \in \mathcal{N}}\chi_{n,m}(\Upsilon-\alpha_n) (\omega_t T_{n,m}^{cmp}+\omega_e E_{n,m}^{cmp}),\tag{\ref{sub-problem1}}\\
    &~~~~~\text{s.t.~} 1 \leq \alpha_n \leq \Upsilon, \forall n \in \mathcal{N},\label{15a}\\
    &~~~~~~~~~~~\text{(\ref{P1_C_f_n}),~(\ref{P1_C_fnm_2})}.\nonumber    
\end{align}
\end{subequations}
The challenge in solving Sub-problem \uppercase\expandafter{\romannumeral1} arises primarily from the multiplication between $\boldsymbol{\alpha,f^V}$ and $\boldsymbol{\alpha,f^R}$.
This multiplication creates a non-linear interaction between the variables, making the problem inherently non-convex. 
To tackle this problem, in this paper, we utilize Sequential quadratic programming (SQP) technique.
SQP is an iterative method typically used for solving constrained non-linear optimization problems~\cite{boggs1995sequential}.
It can also be effective for certain non-convex problems. 
The key idea behind SQP involves approximating the non-linear problem by a quadratic programming subproblem at each iteration, solving this subproblem, and then updating the solution. 

We begin by defining the Lagrangian function for Sub-problem \uppercase\expandafter{\romannumeral1}:

\begin{align}
    &\mathcal{L}(\boldsymbol{\alpha,f^V,f^R,\lambda,\mu,\gamma,\sigma})=\!\!\sum_{n \in \mathcal{N}} \alpha_n \!\cdot(\omega_t T_n^{cmp} \!+\! \omega_e E_n^{cmp})+\nonumber\\&
    \sum_{m \in \mathcal{M}}\sum_{n \in \mathcal{N}}\chi_{n,m}(\Upsilon-\alpha_n) (\omega_t T_{n,m}^{cmp}+\omega_e E_{n,m}^{cmp})+ \nonumber\\&\sum_{n \in \mathcal{N}} \lambda_n (1\!-\!\alpha_n) \!+\!\sum_{n \in \mathcal{N}} \mu_n (\alpha_n\!-\!\Upsilon) \!+ \! \sum_{n \in \mathcal{N}} \gamma_n (f_n\!-\!f_n^{max}) + \nonumber\\&\sum_{m \in \mathcal{M}}\sigma_m ( \sum_{n \in \mathcal{N}}\chi_{n,m}f_{n,m}-f_m^{max}).
\end{align}
Next, the partial derivative of the Lagrangian function is given by:
\begin{align}
    \Delta\mathcal{L} = [\frac{\partial \mathcal{L}}{\alpha_n} ~\frac{\partial \mathcal{L}}{f_n} ~\frac{\partial \mathcal{L}}{f_{n,m}} ~\frac{\partial \mathcal{L}}{\lambda_n} ~\frac{\partial \mathcal{L}}{\mu_n} ~\frac{\partial \mathcal{L}}{\gamma_n} ~\frac{\partial \mathcal{L}}{\sigma_m}]^\intercal,
\end{align}
where
\\[-5pt]
\begin{align}
    &\frac{\partial \mathcal{L}}{\partial \alpha_n} = \omega_tT_n^{cmp}+\omega_eE_n^{cmp} - \nonumber \\ &~~~~\sum_{m \in \mathcal{M}}\chi_{n,m}(\omega_t T_{n,m}^{cmp}+\omega_e E_{n,m}^{cmp})-\lambda_n + \mu_n,\\
    &\frac{\partial \mathcal{L}}{\partial f_n} = -\frac{\alpha_n \omega_t \psi(d_n)}{C_n^V D_n^V f_n^2} + \frac{2\alpha_n \omega_e \kappa_1 f_n \psi(d_n)}{C_n^V D_n^V}+\gamma_n, \\
    &\frac{\partial \mathcal{L}}{\partial f_{n,m}} = \chi_{n,m} (\Upsilon \!-\!\alpha_n) \Big(\frac{2\omega_e \kappa_2 f_{n,m}\psi(d_n)}{C_m^R D_m^R}-\frac{\omega_t \psi(d_n)}{C_m^R D_m^R f_{n,m}^2}\Big) \nonumber \\
    &~~~~~~~~~~~~~+ \sigma_m \chi_{n,m},\\
    &\frac{\partial \mathcal{L}}{\partial \lambda_n} = 1-\alpha_n,\\
    &\frac{\partial \mathcal{L}}{\partial \mu_n} = \alpha_n- \Upsilon,\\
    &\frac{\partial \mathcal{L}}{\partial \gamma_n} = f_n - f_n^{max},\\
    &\frac{\partial \mathcal{L}}{\partial \sigma_m} = \sum_{n \in \mathcal{N}}\chi_{n,m}f_{n,m}-f_m^{max}.
\end{align}
The Hessian of the Lagrangian is calculated as (\ref{Hessian}), where
\begin{align}
    &\frac{\partial^2 \mathcal{L}}{\partial \alpha_n^2} = 0,~\frac{\partial^2 \mathcal{L}}{\partial \alpha_n \partial \lambda_n} = -1,~\frac{\partial^2 \mathcal{L}}{\partial \alpha_n \partial \mu_n} = 1,\\
    &\frac{\partial^2 \mathcal{L}}{\partial \alpha_n \partial f_n} = -\frac{ \omega_t \psi(d_n)}{C_n^V D_n^V f_n^2} + \frac{2 \omega_e \kappa_1 f_n \psi(d_n)}{C_n^V D_n^V},\\
    &\frac{\partial^2 \mathcal{L}}{\partial \alpha_n \partial f_{n,m}} = - \chi_{n,m}\Big(\frac{2\omega_e \kappa_2 f_{n,m}\psi(d_n)}{C_m^R D_m^R}-\frac{\omega_t \psi(d_n)}{C_m^R D_m^R f_{n,m}^2}\Big),\\
      &\frac{\partial^2 \mathcal{L}}{\partial \alpha_n \partial \gamma_n} = \frac{\partial^2 \mathcal{L}}{\partial \alpha_n \partial \sigma_m} = 0,\\
      &\frac{\partial^2 \mathcal{L}}{\partial f_n^2} = \frac{2\alpha_n \omega_t \psi(d_n)}{C_n^V D_n^V f_n^3} \!+\! \frac{2\alpha_n \omega_e \kappa_1 \psi(d_n)}{C_n^V D_n^V}, \\
      &\frac{\partial^2 \mathcal{L}}{\partial f_n \partial \gamma_n} = 1,\\
      &\frac{\partial^2 \mathcal{L}}{\partial f_n \partial f_{n,m}} = \frac{\partial^2 \mathcal{L}}{\partial f_n \partial \lambda_n} = \frac{\partial^2 \mathcal{L}}{\partial f_n \partial \mu_n} = \frac{\partial^2 \mathcal{L}}{\partial f_n \partial \sigma_m} = 0,\\
      &\frac{\partial^2 \mathcal{L}}{\partial f_{n,m}^2} = \chi_{n,m} (\Upsilon -\alpha_n) \Big(\frac{2\omega_e \kappa_2 \psi(d_n)}{C_m^R D_m^R}+\frac{2\omega_t \psi(d_n)}{C_m^R D_m^R f_{n,m}^3}\Big), \\
    &\frac{\partial^2 \mathcal{L}}{\partial f_{n,m} \partial \sigma_m} = \chi_{n,m},\\
    &\frac{\partial^2 \mathcal{L}}{\partial f_{n,m} \partial \lambda_n} = \frac{\partial^2 \mathcal{L}}{\partial f_{n,m} \partial \mu_n} = \frac{\partial^2 \mathcal{L}}{\partial f_{n,m} \partial \gamma_n} = 0,\\
    &\frac{\partial^2 \mathcal{L}}{\partial \lambda_n^2} =\frac{\partial^2 \mathcal{L}}{\partial \mu_n^2} = \frac{\partial^2 \mathcal{L}}{\partial \gamma_n^2} = \frac{\partial^2 \mathcal{L}}{\partial \sigma_m^2} = 0,\\
    &\frac{\partial^2 \mathcal{L}}{\partial \lambda_n \partial \mu_n} = \frac{\partial^2 \mathcal{L}}{\partial \lambda_n \partial \gamma_n} =\frac{\partial^2 \mathcal{L}}{\partial \lambda_n \partial \sigma_m} =0,\\
    &\frac{\partial^2 \mathcal{L}}{\partial \mu_n \partial \gamma_n} =\frac{\partial^2 \mathcal{L}}{\partial \mu_n \partial \sigma_m} = \frac{\partial^2 \mathcal{L}}{\partial \gamma_n \partial \sigma_m} = 0.\\[-13pt]\nonumber
\end{align}
\begin{figure*}
\begin{small}
\begin{align}
 \Delta^2\mathcal{L} = 
        &\left[ \begin{array}{ccccccc}
            \frac{\partial^2 \mathcal{L}}{\partial \alpha_n^2} & \frac{\partial^2 \mathcal{L}}{\partial \alpha_n \partial f_n} & \frac{\partial^2 \mathcal{L}}{\partial \alpha_n \partial f_{n,m}}
            & \frac{\partial^2 \mathcal{L}}{\partial \alpha_n \partial \lambda_n}
            & \frac{\partial^2 \mathcal{L}}{\partial \alpha_n \partial \mu_n}
            & \frac{\partial^2 \mathcal{L}}{\partial \alpha_n \partial \gamma_n}
            &\frac{\partial^2 \mathcal{L}}{\partial \alpha_n \partial \sigma_m}\\[7pt]
            \frac{\partial^2 \mathcal{L}}{\partial f_n \partial \alpha_n} & \frac{\partial^2 \mathcal{L}}{\partial  f_n^2} & \frac{\partial^2 \mathcal{L}}{\partial f_n \partial f_{n,m}}
            & \frac{\partial^2 \mathcal{L}}{\partial f_n \partial \lambda_n}
            & \frac{\partial^2 \mathcal{L}}{\partial f_n \partial \mu_n}
            & \frac{\partial^2 \mathcal{L}}{\partial f_n \partial \gamma_n}
            &\frac{\partial^2 \mathcal{L}}{\partial f_n \partial \sigma_m}\\[7pt]
            \frac{\partial^2 \mathcal{L}}{\partial f_{n,m} \partial \alpha_n } & \frac{\partial^2 \mathcal{L}}{\partial f_{n,m} \partial f_n} & \frac{\partial^2 \mathcal{L}}{\partial f_{n,m}^2}
            & \frac{\partial^2 \mathcal{L}}{\partial f_{n,m} \partial \lambda_n}
            & \frac{\partial^2 \mathcal{L}}{\partial f_{n,m} \partial \mu_n}
            & \frac{\partial^2 \mathcal{L}}{\partial f_{n,m} \partial \gamma_n}
            &\frac{\partial^2 \mathcal{L}}{\partial f_{n,m} \partial \sigma_m}
            \\[7pt]
            \frac{\partial^2 \mathcal{L}}{\partial \lambda_n \partial \alpha_n} & \frac{\partial^2 \mathcal{L}}{\partial \lambda_n \partial f_n} & \frac{\partial^2 \mathcal{L}}{\partial \lambda_n \partial f_{n,m}}
            & \frac{\partial^2 \mathcal{L}}{\partial \lambda_n^2}
            & \frac{\partial^2 \mathcal{L}}{\partial \lambda_n \partial \mu_n}
            & \frac{\partial^2 \mathcal{L}}{\partial \lambda_n \partial \gamma_n}
            &\frac{\partial^2 \mathcal{L}}{\partial \lambda_n \partial \sigma_m}
            \\[7pt]
            \frac{\partial^2 \mathcal{L}}{\partial \mu_n \partial \alpha_n} & \frac{\partial^2 \mathcal{L}}{\partial \mu_n \partial f_n} & \frac{\partial^2 \mathcal{L}}{\partial \mu_n \partial f_{n,m}}
            & \frac{\partial^2 \mathcal{L}}{\partial \mu_n \partial \lambda_n}
            & \frac{\partial^2 \mathcal{L}}{\partial \mu_n^2}
            & \frac{\partial^2 \mathcal{L}}{\partial \mu_n \partial \gamma_n}
            &\frac{\partial^2 \mathcal{L}}{\partial \mu_n \partial \sigma_m}
            \\[7pt]
            \frac{\partial^2 \mathcal{L}}{\partial \gamma_n \partial \alpha_n} & \frac{\partial^2 \mathcal{L}}{\partial  \gamma_n \partial f_n} & \frac{\partial^2 \mathcal{L}}{\partial \gamma_n \partial f_{n,m}}
            & \frac{\partial^2 \mathcal{L}}{\partial \gamma_n \partial \lambda_n}
            & \frac{\partial^2 \mathcal{L}}{\partial \gamma_n \partial \mu_n}
            & \frac{\partial^2 \mathcal{L}}{\partial \gamma_n^2}
            &\frac{\partial^2 \mathcal{L}}{\partial \gamma_n \partial \sigma_m}
            \\[7pt]
            \frac{\partial^2 \mathcal{L}}{\partial \sigma_m \partial \alpha_n} & \frac{\partial^2 \mathcal{L}}{\partial  \sigma_m \partial f_n} &\frac{\partial^2 \mathcal{L}}{\partial \sigma_m \partial f_{n,m}}
            & \frac{\partial^2 \mathcal{L}}{\partial \sigma_m \partial \lambda_n}
            & \frac{\partial^2 \mathcal{L}}{\partial \sigma_m \partial \mu_n}
            & \frac{\partial^2 \mathcal{L}}{\partial \sigma_m \partial \gamma_n}
            &\frac{\partial^2 \mathcal{L}}{\partial \sigma_m^2}
            \end{array} 
        \right ].\label{Hessian}  \\[-26pt]\nonumber
\end{align}
\end{small}
\end{figure*}
Let $g_{n}(\alpha_n)=1-\alpha_n$, $h_{n}(\alpha_n)=\alpha_n-\Upsilon$, $k_{n}(f_n)=f_n - f_n^{max}$, and $q_{m}(f_{n,m})=\sum_{n \in \mathcal{N}}\chi_{n,m}f_{n,m}-f_m^{max}$.
Then, we form a quadratic approximation of (\ref{sub-problem1}) and linearize the constraints in (\ref{P1_C_f_n}),~(\ref{P1_C_fnm_2}) and (\ref{15a}) as follows:
\begin{subequations}\label{SQP_problem}
\begin{align}
    &\min_{\boldsymbol{d}} F(\boldsymbol{\alpha}^{(t)},\boldsymbol{f^V}^{(t)},\boldsymbol{f^R}^{(t)}) + \Delta F(\boldsymbol{\alpha}^{(t)},\boldsymbol{f^V}^{(t)},\boldsymbol{f^R}^{(t)})^\intercal \boldsymbol{d} \nonumber\\
    &+\frac{1}{2} \boldsymbol{d}^\intercal \Delta^2 \mathcal{L}(\boldsymbol{\alpha}^{(t)},\boldsymbol{f^V}^{(t)}\!,\boldsymbol{f^R}^{(t)}\!,\boldsymbol{\lambda}^{(t)}\!,\boldsymbol{\mu}^{(t)}\!,\boldsymbol{\gamma}^{(t)}\!,\boldsymbol{\sigma}^{(t)})\boldsymbol{d},\tag{\ref{SQP_problem}}\\
    &~~~~\text{s.t.~} \Delta g_n(\alpha_n^{(t)}) \boldsymbol{d} + g_n(\alpha_n^{(t)}) \leq 0,\\
    &~~~~\phantom{s.t.~}\Delta h_n(\alpha_n^{(t)}) \boldsymbol{d} + h_n(\alpha_n^{(t)}) \leq 0,\\
    &~~~~\phantom{s.t.~}\Delta k_n(f_n^{(t)}) \boldsymbol{d} + k_n(f_n^{(t)}) \leq 0,\\
    &~~~~\phantom{s.t.~}\Delta q_m(f_{n,m}^{(t)}) \boldsymbol{d} + q_m(f_{n,m}^{(t)}) \leq 0.
\end{align}
\end{subequations}
It is clear that this problem is a convex optimization problem and can be easily solved by standard convex solvers such as CVX~\cite{grant2014cvx}.
After obtaining the step direction $\boldsymbol{d}^{(t)}$, we update the variables as:
\begin{align}
    &\boldsymbol{\alpha}^{(t+1)} = \boldsymbol{\alpha}^{(t)} + \boldsymbol{d}^{(t)}, \boldsymbol{f^V}^{(t+1)} = \boldsymbol{f^V}^{(t)} + \boldsymbol{d}^{(t)},\\
    &\boldsymbol{f^R}^{(t+1)} = \boldsymbol{f^R}^{(t)} + \boldsymbol{d}^{(t)}, \boldsymbol{\lambda}^{(t+1)} = \boldsymbol{\lambda}^{(t)} + \boldsymbol{d}^{(t)},\\
    &\boldsymbol{\mu}^{(t+1)} = \boldsymbol{\mu}^{(t)} + \boldsymbol{d}^{(t)}, \boldsymbol{\gamma}^{(t+1)} =\boldsymbol{\gamma}^{(t)} + \boldsymbol{d}^{(t)},\\
    &\boldsymbol{\sigma}^{(t+1)} =\boldsymbol{\sigma}^{(t)} + \boldsymbol{d}^{(t)}.
\end{align}
 This iterative process involves repeatedly adjusting the variables in the direction of $\boldsymbol{d}^{(t)}$, which is calculated to move towards the optimal solution. With each iteration, the variables are refined, gradually leading the system closer to the desired state of convergence.
\subsection{Sub-problem \uppercase\expandafter{\romannumeral2}}
With fixed $\boldsymbol{\alpha,f^V,f^R}$, problem $\mathbb{P}_1$ can be rewritten as:
\begin{align}
    \text{Sub-problem \uppercase\expandafter{\romannumeral2}:}~&\min_{\boldsymbol{p,b}} \sum_{n \in \mathcal{N}} \omega_e E_{n}^{com},\\
    &\text{s.t.~(\ref{P1_C_pn}),~(\ref{P1_C_bn_2})}.\nonumber    
\end{align}
This optimization problem presents itself as a sum-of-ratios minimization problem, characterized by its intrinsic complexity and non-linearity. 
Owing to its NP-hard nature, the problem poses significant computational difficulties, rendering conventional optimization methodologies largely ineffective.
To handle this problem, we transform it into an epigraph form by introducing an auxiliary variable $\beta_{n}$.
Specifically, we let 
\begin{align}
 \frac{p_nd_n}{\sum_{m \in \mathcal{M}} \chi_{n,m} r_{n,m}} \leq \beta_n. \label{beta}\\[-13pt]\nonumber
\end{align}
Then, Sub-problem \uppercase\expandafter{\romannumeral2} is equivalent to the following optimization problem:
\begin{subequations}\label{SP1_v2}
\begin{align}
    &\min_{\boldsymbol{p,b,\beta}} \sum_{n \in \mathcal{N}} \omega_e \beta_{n} \tag{\ref{SP1_v2}},\\
    & \text{s.t.}~p_n d_n-\!\sum_{m \in \mathcal{M}}\!\chi_{n,m} r_{n,m}\beta_n  \leq 0, \label{C_beta} \\
    &~~~~~\text{(\ref{P1_C_pn}),~(\ref{P1_C_bn_2})},\nonumber
\end{align}
\end{subequations}
where constraint (\ref{C_beta}) comes from (\ref{beta}).
The problem (\ref{SP1_v2}) is still non-convex and difficult to handle due to constraint (\ref{C_beta}).
To tackle the above problem, we next give the following lemma:
\begin{lem}
Let $R(p_n,b_n)=\sum_{m \in \mathcal{M}} \chi_{n,m} r_{n,m}$.
Then, $R(p_n,b_n)$ is concave with respect to $p_n$ and $b_n$.
\end{lem}
\begin{proof}
the Hessian matrix of $R(p_n,b_n)$ can be given by:
\begin{align}
    &Hessian(R) = \nonumber\\
        &\left[ \begin{array}{cc}
            -\frac{\sum_{m \in \mathcal{M}} \chi_{n,m} g_{n,m}^2}{\ln2 \cdot b_n \sigma^4(1+\frac{g_{n,m}p_n}{\sigma^2 b_n })^2} & \frac{\sum_{m \in \mathcal{M}} \chi_{n,m} g_{n,m}^2p_n}{\ln2 \cdot (\sigma^2 b_n+g_{n,m}p_n)^2} \\
            \frac{\sum_{m \in \mathcal{M}} \chi_{n,m} g_{n,m}^2p_n}{\ln2 \cdot (\sigma^2 b_n+g_{n,m}p_n)^2} & -\frac{\sum_{m \in \mathcal{M}} \chi_{n,m} g_{n,m}^2p_n^2}{\ln2 \cdot b_n (\sigma^2 b_n + g_{n,m}p_n)^2} 
            \end{array} 
        \right ].\\[-5pt] \nonumber
\end{align}
Considering any vector $\boldsymbol{x}$ represented as $\boldsymbol{x} = [x_1,x_2]^\intercal$ belonging to the two-dimensional real space $\mathbb{R}^2$, one can note the following observations:
\begin{align}
    \boldsymbol{x}^\intercal Hessian(R) \boldsymbol{x} = - \frac{\sum_{m \in \mathcal{M}} \chi_{n,m} g_{n,m}^2 (b_n x_1 \!-\! p_n x_2)^2}{\ln2 \cdot b_n (b_nN_0\!+\!g_{n,m}p_n)^2} \leq 0.\nonumber
\end{align}
Hence, $Hessian(R)$ is classified as a negative semidefinite matrix.
Consequently, $R(p_n,b_n)$ exhibits concavity with respect to both $p_n$ and $b_n$.
\end{proof}
The primary objective function of Sub-problem \uppercase\expandafter{\romannumeral2} consists of a series of fractional functions. 
It is evident that the product $p_n d_n$ is convex, 
and as stated in Lemma 1, $R(p_n,b_n)$ is concave with respect to both $p_n$ and $b_n$.
Owing to these characteristics, problem (\ref{SP1_v2}) can be reformulated into a problem with a subtractive form based on the next theorem.
\begin{thm}
If $[\boldsymbol{p^*,b^*,\beta^*}]$ is a globally optimal solution to problem (\ref{SP1_v2}), then
\begin{align}
    \beta_n^* = \frac{p_n^*d_n}{\sum_{m \in \mathcal{M}} \chi_{n,m} \cdot R(p_n^*,b_n^*)},\label{beta}
\end{align}
and there exists 
\begin{align}
    \nu = \frac{\omega_e}{R(p_n^*,b_n^*)}, \label{nu}
\end{align}
such that $[\boldsymbol{p^*,b^*}]$ is the solution to the following optimization problem where $\beta_n=\beta^*_n$:       
\begin{subequations}\label{SP1_v3}
\begin{align}
&\min_{\boldsymbol{p,b}} \sum_{n \in \mathcal{N}} \nu_n \Big(p_n d_n -\beta_{n} \sum_{m \in \mathcal{M}} \chi_{n,m} R(p_n,b_n)\Big)\tag{\ref{SP1_v3}},\\
& \text{s.t.~(\ref{P1_C_pn}),~(\ref{P1_C_bn_2})}.\nonumber
\end{align}
\end{subequations}
\end{thm}
\begin{proof}
We can obtain the proof after applying Lemma 2.1 of~\cite{jong2012efficient} to the optimization problem in (\ref{SP1_v2}).
\end{proof}
Theorem 1 clearly establishes that problem (\ref{SP1_v3}) has an optimal solution $[\boldsymbol{p^*,b^*}]$ that is consistent with the solution of problem (\ref{SP1_v2}).
As a result, an initial step involves providing values for $[\boldsymbol{\mu,\beta}]$ to find a solution $[\boldsymbol{p,b}]$ by resolving the problem (\ref{SP1_v3}). Following this, the subsequent step is to determine $[\boldsymbol{\mu,\beta}]$ as per equations (\ref{beta}) and (\ref{nu}). 
This process can be conceptualized as having two layers: the first step acts as an internal cycle dedicated to addressing problem problem (\ref{SP1_v3}), while the second step serves as the external cycle, focused on pinpointing the optimal values of $[\boldsymbol{\mu,\beta}]$.
The sum $\sum_{m \in \mathcal{M}} \chi_{n,m} R(p_n,b_n)$ being concave with respect to both $p_n$ and $b_n$ leads to the objective function of problem (\ref{SP1_v3}) being convex. 
Additionally, the constraints specified in (\ref{P1_C_pn}) and (\ref{P1_C_bn_2}) are all convex in nature. 
Therefore, this categorizes problem (\ref{SP1_v3}) as a convex optimization problem, and we can use standard convex solvers such as CVX~\cite{grant2014cvx} to solve it.
\vspace{-5pt}\section{Simulation}
In this section, we present the simulation results.
Firstly, we introduce the simulation settings.
Next, we compare the proposed algorithm with other schemes.

\subsection{Simulation Settings}
we investigate a network setup covering a 1000 meters by 1000 meters area, which includes 20 vehicles and 5 Road Side Units (RSUs). 
This configuration is designed to represent a typical urban vehicular network environment. The wireless channel's path loss within this network is calculated using the formula \mbox{$128.1 + 37.6\log_{10}(\texttt{distance})$}, where \texttt{distance} refers to the Euclidean distance between each vehicle and the nearest RSU. 
Additionally, our model incorporates a standard deviation of 8 dB for shadow fading.
The Gaussian noise power $\sigma^2$ is -134dBm.
The total bandwidth for each RSU $b_m^{max}$ is 20 MHz.
The maximum transmission power for vehicles, denoted as $p_n^{max}$, is set at 20 Watts.

\subsection{Comparison with Other Schemes}
In this section, we compare the proposed method with two other schemes:
\begin{enumerate}
    \item \textbf{\underline{R}andom \underline{L}LM offloading with \underline{O}ptimized \underline{R}esource allocation (RLOR).} In this scheme, the vehicles offload random LLM tasks to RSUs, while the computation and communication resource allocation is optimized. 
    \item \textbf{\underline{R}andom \underline{R}esource allocation with \underline{O}ptimized \underline{L}LM offloading (RROL).}  In this scheme, the LLM offloading decision is optimized, while the computation and communication resources are randomly allocated.
\end{enumerate}


Figure \ref{fig_time_w} illustrates the impact of varying weighting factors for delay on the completion time of tasks within the network. 
As the weighting factor $\omega_t$ increases, there is a noticeable decrease in the completion time across all three methods evaluated in the study. 
This trend is a direct consequence of the system placing greater emphasis on minimizing completion time in its operational priorities.
Among the three methods, the performance of RLOR and RROL approaches are worst because they adopt random LLM offloading decision or random resource allocation scheme.
The proposed method outperforms RLOR and RROL in the completion time, demonstrating the effectiveness of the proposed algorithm.

\begin{figure}[tbp]
\centering
\includegraphics[width=0.43\textwidth]{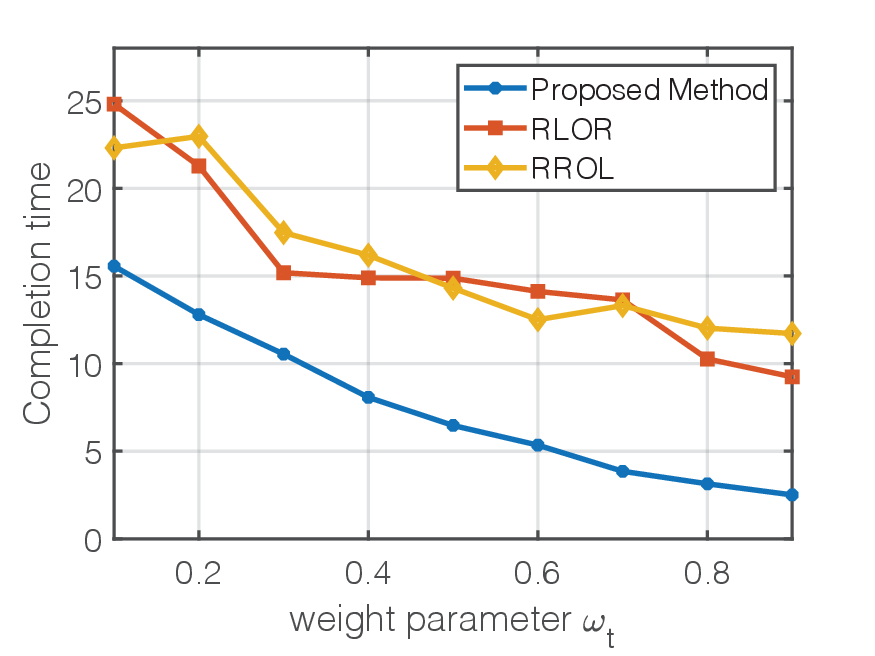}
\vspace{-8pt}\caption{The completion time under different delay weighting factor.}
\label{fig_time_w}\vspace{-15pt}
\end{figure}

\begin{figure}[tbp]
\centering
\includegraphics[width=0.43\textwidth]{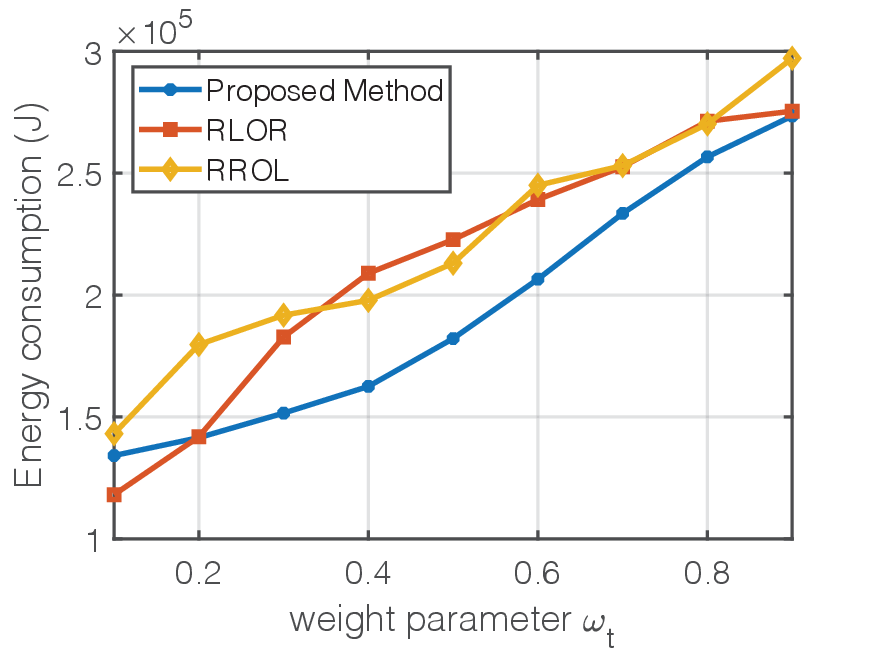}
\vspace{-8pt}\caption{The energy consumption under different delay weighting factor.}
\label{fig_energy_w}\vspace{-15pt}
\end{figure}

In Figure \ref{fig_energy_w}, we present a detailed analysis of the energy consumption under varying weighting factors.
As the weighting factor $\omega_t$ increases, the emphasis within the system shifts more towards minimizing delay, consequently reducing the priority given to energy efficiency. 
This shift in focus is reflected in the increased energy consumption observed across all three methods being compared.
Although in Figure \ref{fig_energy_w}, when $\omega_t$ is from 0.1 to 0.2, the energy consumption of RLOR is almost the same with or even less than the proposed method, the proposed method reduces the completion time by approximately 40\% than RLOR.
The results clearly demonstrate that the proposed method exhibits superior performance in terms of energy consumption when compared to the RLOR and RROL methods.
This noteworthy outcome underscores the effectiveness of the algorithm we proposed. 
The proposed method's ability to minimize energy consumption, while maintaining a short completion time can help to achieve an efficient vehicular communication system.

\vspace{-2pt}\section{Conclusion}\vspace{-1pt}
In this paper, we investigated the integration of vehicular technology with LLMs in 6G vehicular networks.
To address the computational challenges posed by LLMs, we proposed a novel solution where vehicles process the initial few layers of LLM computations locally, combined with offloading more demanding tasks through edge computing for more efficient operation in 6G environments.
We identified and formulated a multi-objective optimization problem, aiming to balance the computational demands of LLMs with the resource constraints inherent in vehicular networks. These are crucial considerations as we move towards the 6G era.
To effectively address this problem, we divided it into two subproblems, allowing for a more efficient solution.
Simulation results demonstrated a significant improvement of the proposed algorithm in resource allocation and energy efficiency, showcasing the potential of LLMs in enhancing vehicular network capabilities. 

\vspace{-2pt}\section*{Acknowledgement}\vspace{-1pt}

This research is supported partly by the Singapore Ministry of Education Academic Research Fund under Grant Tier 1 RT5/23, Grant Tier 1 RG90/22, Grant Tier 1 RG97/20, Grant Tier 1 RG24/20 and Grant Tier 2 MOE2019-T2-1-176; partly by the Nanyang Technological University (NTU)-Wallenberg AI, Autonomous Systems and Software Program (WASP) Joint Project; and partly by Imperial-Nanyang Technological University Collaboration Fund INCF-2024-008.

\bibliographystyle{IEEEtran}
\bibliography{related}

\end{document}